\pdfoutput=1
\documentclass[letterpaper,10pt,conference]{ieeeconf}

\IEEEoverridecommandlockouts
\usepackage[T1]{fontenc}

\usepackage{amsmath,amssymb,amsthm}
\usepackage[noend,linesnumbered]{algorithm2e}
\usepackage{graphicx}
\usepackage{cite}
\usepackage{url}
\usepackage{enumerate}
\usepackage{xcolor}
\usepackage{booktabs}
\usepackage{tikz}
\usetikzlibrary{positioning}
\usepackage{todonotes}

\newtheorem{proposition}{Proposition}
\newtheorem{lemma}{Lemma}
\newtheorem{corollary}{Corollary}
\newtheorem{remark}{Remark}

\newtheorem{problem}{Problem}

\newcommand{\Expect}{\mathbb{E}}
\newcommand{\Prob}{\mathbb{P}}
\newcommand{\Real}{\mathbb{R}}
\newcommand{\cS}{\mathcal{S}}
\newcommand{\cA}{\mathcal{A}}
\newcommand{\cT}{\mathcal{T}}
\newcommand{\cU}{\mathcal{U}}
\newcommand{\cF}{\mathcal{F}}
\newcommand{\cR}{\mathcal{R}}

\newcommand{\dtheta}{\Delta\theta}

\begin{document}

\title{\LARGE \bf
Policy Synthesis for Finite Populations of MDP Agents under Aggregate Reach-Avoid Chance Constraints}

\author{Jie Fu$^{1}$ and Anamika Dubey$^{2}$%
\thanks{$^{1}$Jie Fu is with the Department of Electrical and Computer
Engineering, University of Florida, Gainesville, FL 32611, USA
{\tt\small fujie@ufl.edu}.}%
\thanks{$^{2}$Anamika Dubey is with the School of Electrical Engineering
and Computer Science, Washington State University, Pullman, WA 99164, USA
{\tt\small anamika.dubey@wsu.edu}.}%
}

\maketitle
\thispagestyle{empty}
\pagestyle{empty}

\begin{abstract}
Consider a finite population of agents with decoupled Markov transition dynamics and empirical-density feedback, subject to the following constraints: with probability at least $1-\delta_r$, at least a fraction $\alpha_r$ of agents must reach a target region at some time $t^*$, while, at each time up to $t^*$, the unsafe population fraction must remain below $\beta_u$ with probability at least $1-\delta_u$. However,
standard mean-field methods enforce these constraints only in
expectation, which fails to account for stochastic fluctuations at
finite fleet size $N$. To address this control problem, we propagate the second-order moment (variance) of the empirical
density alongside the mean-field trajectory via a discrete-time
Lyapunov recursion, and apply the Cantelli inequality to convert
chance constraints into tractable deterministic conditions on the moments of the empirical density.
We then incorporate these moment-based surrogate constraints into a gradient-based sequential convex approximation procedure for density-feedback policy synthesis. We further introduce additional   moment-error bounds  to construct a rigorous finite-$N$ certificate.
The method is evaluated on a gridworld environment and a power-system EV-charging aggregation problem and compared with a standard deterministic population-level LP baseline.
\end{abstract}


\section{Introduction}
\label{sec:intro}

Controlling large populations of autonomous agents subject to safety
and mission requirements is a fundamental challenge in swarm
robotics~\cite{yang2024rover,doerrControlLargeSwarms2018,demir2015decentralized},
multi-agent systems~\cite{elchamie2016convex}, and
smart-grid applications~\cite{couilletMeanFieldGameElectric2012}.
Tracking the state of every agent individually becomes intractable as
$N$ grows because the resulting joint state space scales exponentially
with $N$; mean-field methods instead describe the population through
an aggregate distribution whose dimension is independent of
$N$~\cite{acikmeseMarkovChainApproach2012a,lasryMeanFieldGames2007,huangLargePopulationStochastic2006}.

For finite $N$, however, the empirical population density exhibits
$O(1/\sqrt{N})$ stochastic fluctuations around its mean-field limit,
so a policy whose limiting trajectory satisfies a constraint may still
violate that constraint with non-negligible probability. We therefore
study policy synthesis for a finite-$N$ population subject to aggregate
reach--avoid chance constraints: with probability at least $1-\delta_r$,
at least a fraction $\alpha_r$ of the agents must occupy the target
region at some time $t^*$, while, at every time up to $t^*$, the
unsafe-region occupancy must not exceed $\beta_u$ with probability at
least $1-\delta_u$.

 \noindent \textbf{Related work.}
 Mean-field games study equilibria among noncooperative agents, whereas
 mean-field control (MFC) considers a cooperative population whose
 aggregate objective is optimized by a centralized planner. Foundational
 continuous-time mean-field game theory characterizes population
 equilibria through coupled backward optimality and forward
 population-dynamics equations
 \cite{lasryMeanFieldGames2007,huangLargePopulationStochastic2006}. For
 discrete-time mean-field MDPs, Bäuerle establishes the existence of
 optimal policies and convergence of finite-population control problems
 to their mean-field limits~\cite{bauerleMeanFieldMarkov2023a}. 
 Related mean-field reinforcement learning (MFRL) studies develop and
 analyze learning algorithms for both mean-field games and cooperative
 mean-field control
\cite{yangMeanFieldMultiAgent2018,fouqueLauriereZhangConvergence2025,angiuliUnifiedReinforcement2022,carmonaLauriereTanLinearQuadratic2019}.  

For the class of reach-avoid  chance constraints, 
stochastic reach--avoid methods has been developed to ensure probabilistic
specifications for individual stochastic systems~\cite{zikelic2023}. Yet, the method does not extend directly to enforce  constraints on the empirical density of a finite
multi-agent population when individual agents are stochastic systems. 
For MFC, 
 safety constraints have been   imposed on the  deterministic population distribution or the system's mean-field limit ~\cite{acikmeseMarkovChainApproach2012a,elchamie2016convex,yang2024rover,mfcbf2024}. Safe MFRL method~\cite{jusupSafeM3UCRL2023}
accounts for uncertainty in unknown transition dynamics when learning
policies subject to global population constraints, but it does not
account for stochastic fluctuations of the realized empirical density
at finite $N$.

Relating the proposed problem to other work focusing on stochastic variability in control design.
Risk-sensitive mean-field games~\cite{saldiApproximateMarkovNash2020}
use an exponential objective to account for variability in an
individual agent's accumulated return, rather than fluctuations in the
aggregate occupancy of a finite population. Covariance steering uses
the mean and covariance of a stochastic system's state to enforce
trajectory-level chance constraints~\cite{okamotoOptimalCovarianceControl2018},
while distribution- and density-steering methods control  state
distributions under non-Gaussian disturbances or in robotic
swarms~\cite{sivaramakrishnanDistributionSteeringDiscretetime2022,sinigagliaRobustOptimalDensity2022}.
These approaches characterize either noise intrinsic to an
individual system's dynamics, or spatial dispersion of an
continuum-population density. In contrast, the covariance we aim to constrain is the
$O(1/N)$ sampling covariance of the discrete empirical occupancy
vector, which arises only because the population is finite and
vanishes as $N\to\infty$.

Related finite-population work~\cite{bayraktarFiniteApproximations2023}
formulates mean-field-type multi-agent control as a measure-valued MDP
and develops discretization-based approximations with near-optimality
guarantees. This measure-valued state representation avoids the
exponentially large joint state space, but the planning method
discretizes the measure-valued state space into a finite set of
quantized empirical distributions, whose number grows combinatorially
with the population size and the number of quantization bins.
Moreover, the resulting approximate policies carry no chance-constraint
guarantee for the finite-$N$ empirical density.
For uncontrolled stochastic population models, \cite{gastRefinedMeanField2017}
derives finite-size corrections to the mean-field limit at large but
finite $N$. This is done by solving a Lyapunov equation which yields a bias-correction term
of order $1/N$ that sharpens the mean-field estimate of a performance
metric's expected value. This work  does not characterize the
variance of the finite-$N$ population distribution, nor does it
address chance-constrained policy synthesis.

\textbf{Contribution.}
We study a finite population of homogeneous MDP agents with
population-independent transition kernels and density-feedback
policies. The goal is to synthesize a policy that, when each agent follows this policy, their collective behavior optimizes an aggregate objective subject to reach--avoid
chance constraints on the empirical density.
We derive first-order recursions for the mean and covariance of the
finite-population empirical density under density-feedback policies and
use Cantelli's inequality to obtain deterministic sufficient conditions
for the reach--avoid chance constraints. 
Because the mean and covariance recursions are first-order
approximations, we introduce bounds on the
moment-approximation errors to  the
reach and avoid margins to derive rigorous
finite-$N$ certificates.
Using these constraint surrogates,
we develop a model-based Sequential Convex Approximation (SCA)
algorithm that updates the density-feedback policy parameters through
a sequence of second-order cone programs, using automatic
differentiation to linearize the reach--avoid constraints locally.  The method is validated using two experiments. 
\section{Problem Formulation}
\label{sec:problem}

\subsection{Swarm MDP and Empirical Density}

A finite-horizon MDP is $M=(\cS,\cA,P,r,\mu_0)$, where $\cS$ and
$\cA$ are finite state and action spaces, $P(s'\mid s,a)$ is the
transition kernel, $r:\cS\times\cA\to\Real$ is the reward, and
$\mu_0\in\Delta(\cS)$ is the mean initial population density.
A swarm of $N$ homogeneous agents evolves on $M$.
The \emph{empirical density} is
\begin{equation}
    \mu_t^N(s) = \frac{1}{N}\sum_{i=1}^N \mathbf{1}[X_t^i=s],
    \quad s\in\cS,
    \label{eq:empirical}
\end{equation}
where $X_t^i\in\cS$ is agent $i$'s state at time $t$.
We allow $\mu_0^N$ to be random with known mean
$\mu_0=\Expect[\mu_0^N]$ and covariance $\Sigma_0$; deterministic
initialization is recovered by setting $\Sigma_0=0$.
Each agent selects actions via the time-varying density-feedback
policy $\pi_{\theta_t}(a\mid s,\mu_t^N)$, parameterized by
$\theta=(\theta_0,\ldots,\theta_{T-1})
\in\Real^{T\times|\cS|\times|\cA|\times n_{\rm feat}}$, where
$\theta_t\in\Real^{|\cS|\times|\cA|\times n_{\rm feat}}$ is the slice
of parameters used at time $t$ and $n_{\rm feat}$ is the number of
features. Specifically, we use the softmax-linear policy
\begin{equation}
\pi_{\theta_t}(a\mid s,\mu)
=\frac{\exp\!\left(\theta_t[s,a,:]^\top\phi(s,\mu)\right)}
{\sum_{a'\in\cA}\exp\!\left(\theta_t[s,a',:]^\top
\phi(s,\mu)\right)},
\label{eq:policy}
\end{equation}
where $\phi(s,\mu)\in\Real^{n_{\rm feat}}$ is a twice continuously
differentiable feature vector that may include state and
population-density information.
We write $\pi_\theta$ for brevity when the time index is clear from
context.
Because the transition kernel and reward are population-independent,
the density-feedback policy couples the agents through the controlled
transition dynamics.

\begin{remark}[Density-Dependent Reward: General Case]
	\label{rem:density_reward}
	
	The reward $r(s,a)$ in Problem~\ref{prob:main} depends only on an
	agent's state-action pair and not on the population density. More
	general density-dependent rewards arise in congestion models and
	coverage control but are outside the present scope.
\end{remark}

\subsection{Reach-Avoid Chance Constraints}

Let $\cT,\cU\subseteq\cS$ be disjoint target and unsafe regions, and
define the free region as $\cF=\cS\setminus(\cT\cup\cU)$.
Define the aggregate density in region $\cR$ as
$\mu_t^N(\cR)=\sum_{s\in\cR}\mu_t^N(s)$.
Let $[T]:=\{1,\ldots,T\}$ denote the set of time steps.

\begin{problem}[Chance-Constrained Reach-Avoid]
\label{prob:main}
Given $\alpha_r,\beta_u,\delta_r,\delta_u\in(0,1)$, find policy
parameters $\theta^*$ and a reach time $t^*\in[T]$ solving
\begin{subequations}
\label{eq:problem}
\begin{align}
    \max_{\theta,\,t^*\in[T]}\quad &
    J(\theta) = \Expect^{\pi_\theta}\!\Bigl[
        \textstyle\sum_{t=0}^{T-1} R_t
    \Bigr]
    \\
    \text{s.t.}\quad &
    \Prob^{\pi_\theta}\!\bigl(\mu_{t^*}^N(\cT)\ge\alpha_r\bigr)
    \ge 1-\delta_r,
    \label{eq:reach}
    \\
    &\Prob^{\pi_\theta}\!\bigl(\mu_t^N(\cU)\le\beta_u\bigr)
    \ge 1-\delta_u,
    \quad \forall t\le t^*,
    \label{eq:avoid}
\end{align}
\end{subequations}
where $R_t=\sum_{s,a}r(s,a)\mu_t^N(s)
\pi_{\theta_t}(a\mid s,\mu_t^N)$ is the population-average one-step
reward conditional on $\mu_t^N$.
\end{problem}
The constraints in~\eqref{eq:reach}--\eqref{eq:avoid} are   bound the reach and avoid violation
probabilities separately, with the avoid guarantee imposed at each
time $t\le t^*$.

\section{Mean-Field Dynamics and Covariance Propagation}
\label{sec:dynamics}

\subsection{Mean-Field Trajectory}

For a fixed finite horizon and under the regularity and
initial-concentration assumptions stated below, $\mu_t^N$ concentrates
around the deterministic \emph{mean-field trajectory} $\bar\mu_t$ as
$N\to\infty$. This trajectory is defined by the \emph{mean-field flow
equation}~\cite{acikmeseMarkovChainApproach2012a}:
\begin{align}
    \bar\mu_{t+1}(s';\theta)
    &= \sum_{s,a} P(s'\mid s,a)\,
      \pi_{\theta_t}(a\mid s,\bar\mu_t)\,\bar\mu_t(s),
      \label{eq:mf_flow}\\
    \bar\mu_0&=\mu_0.
\end{align}
We distinguish $\bar\mu_t$ from the true finite-$N$ expectation
$m_t^N:=\Expect[\mu_t^N]$. Define the conditional-mean map
\begin{equation}
    g_{\theta_t}(\mu)(s'):=\sum_{s,a}P(s'\mid s,a)\,\pi_{\theta_t}(a\mid s,\mu)\mu(s).
    \label{eq:conditional_mean_map}
\end{equation}
As shown below, the finite-population process satisfies
$\Expect[\mu_{t+1}^N\mid\mu_t^N]=g_{\theta_t}(\mu_t^N)$ exactly,
whereas the deterministic trajectory satisfies
$\bar\mu_{t+1}=g_{\theta_t}(\bar\mu_t)$. By Jensen's inequality,
\begin{equation}
 m_{t+1}^N=\Expect[g_{\theta_t}(\mu_t^N)]
 \neq g_{\theta_t}(m_t^N)
 \label{eq:mean_jensen_gap}
\end{equation}
in general 
when density feedback makes $g_{\theta_t}$ nonlinear. So setting $\bar\mu_t=m_t^N$ on the right-hand side does not
reproduce $m_{t+1}^N$ on the left. 
Accordingly, $\bar\mu_t$ is a nominal deterministic population trajectory
and a leading-order approximation to the finite-$N$ mean, not the exact
expectation.  

Define the time-$t$ \emph{occupation measure} by
\begin{equation}
z_t(s,a;\theta):=\pi_{\theta_t}(a\mid s,\bar\mu_t(\theta))
\bar\mu_t(s;\theta).
\label{eq:occupation_measure}
\end{equation}
For compactness, define
$f_t(\theta,\mu)(s,a):=\pi_{\theta_t}(a\mid s,\mu)\mu(s)$ and
$\mathbf z_t(\theta)=f_t(\theta,\bar\mu_t(\theta))$. The \emph{flow
matrix} $F\in\Real^{|\cS|\times|\cS||\cA|}$ is defined by
$F[s',sa]=P(s'\mid s,a)$.
The mean-field flow can then be written compactly as
\begin{equation}
    \bar\mu_{t+1}(\theta) = F\,\mathbf{z}_t(\theta) = Ff_t(\theta,\bar\mu_t(\theta)),
    \quad \bar\mu_0 = \mu_0.
    \label{eq:mf_recurrence}
\end{equation}

\subsection{Exact Conditional Moments and First-Order Covariance Propagation}

The empirical density itself is a finite-state Markov process under a
density-feedback policy. For an arbitrary density $\mu$, define the
effective one-agent transition matrix
\begin{equation}
 Q_{\theta_t}(\mu)[s',s]
 :=\sum_{a\in\cA}\pi_{\theta_t}(a\mid s,\mu)P(s'\mid s,a).
 \label{eq:Qmu}
\end{equation}
The conditional-mean map in~\eqref{eq:conditional_mean_map} then
satisfies $g_{\theta_t}(\mu)=Q_{\theta_t}(\mu)\mu=Ff_t(\theta,\mu)$.
Conditional on $\mu_t^N=\mu$ for an attainable empirical density
$\mu\in\Delta_N(\cS):=\{\nu\in\Delta(\cS):N\nu(s)\in\mathbb Z_{\ge0}
\ \forall s\}$, there are $N\mu(s)$ agents in state $s$.
Given this count vector, agents randomize their actions independently and
their physical transitions are conditionally independent. Consequently,
\begin{equation}
 \Expect[\mu_{t+1}^N\mid\mu_t^N=\mu]
 =g_{\theta_t}(\mu)
 \label{eq:exact_cond_mean}
\end{equation}
exactly. Define
\begin{align} q_s(\mu)&:=Q_{\theta_t}(\mu)[\cdot,s],\\
 \Lambda_{\theta_t}(\mu)
 &:=
 \sum_{s\in\cS}\mu(s)
 \left[\operatorname{Diag}(q_s(\mu))
 -q_s(\mu)q_s(\mu)^\top\right] 
 \label{eq:Lambda}. 
\end{align}
Summing the covariances of the conditionally independent categorical
next-state indicators gives the exact identity
\begin{equation}
 \operatorname{Cov}(\mu_{t+1}^N\mid\mu_t^N=\mu)
 =\frac{1}{N}\Lambda_{\theta_t}(\mu).
 \label{eq:exact_cond_cov}
\end{equation}

For density feedback, fluctuations of $\mu_t^N$ propagate through the
Jacobian of $g_{\theta_t}$. Let
\begin{equation}
 H_t(\theta):=
 \left.\frac{\partial g_{\theta_t}(\mu)}{\partial\mu}
 \right|_{\bar\mu_t}
 =FJ_t^\mu(\theta),
 \label{eq:Amu_def}
\end{equation}
where, coordinatewise,
\begin{multline}
 J_t^\mu(s,a)[s'']
 =\pi_{\theta_t}(a\mid s,\bar\mu_t)\delta_{ss''}
 {}+\bar\mu_t(s)\\
 \left.\partial_{\mu(s'')}\pi_{\theta_t}(a\mid s,\mu)
 \right|_{\bar\mu_t}.
 \label{eq:Jmu_coord_early}
\end{multline}
When the policy is density independent, the second term vanishes and
$H_t (\theta)=Q_{\theta_t}(\bar\mu_t)$.

We say that the initial empirical density satisfies the \emph{initial
concentration condition} if
\begin{equation}
\mu_0^N-\mu_0=O_p(N^{-1/2}).
\label{eq:initial_concentration}
\end{equation}
This condition holds for deterministic initialization and for
independently sampled initial states; in the latter case,
\begin{equation}
\Expect[\|\mu_0^N-\mu_0\|^2]
=\frac{1-\|\mu_0\|_2^2}{N}.
\label{eq:initial_mean_square}
\end{equation}

\begin{proposition}[First-Order Covariance Approximation]
\label{prop:lyapunov}
Assume a fixed horizon, the initial concentration condition, and that
$g_{\theta_t}$ is twice continuously differentiable in a neighborhood
of the nominal trajectory. Linearizing the conditional-mean map around
the nominal trajectory gives the following first-order approximation
$\widehat\Sigma_t$ to
$\Sigma_t:=\operatorname{Cov}(\mu_t^N)$:
\begin{equation}
 \widehat\Sigma_{t+1}
 =H_t\widehat\Sigma_tH_t^\top
 +\frac{1}{N}\Lambda_{\theta_t}(\bar\mu_t).
 \label{eq:lyapunov}
\end{equation}
The recursion is initialized with $\widehat\Sigma_0=\Sigma_0$,
including $\widehat\Sigma_0=0$ for deterministic initialization.
\end{proposition}

\begin{proof}[Proof sketch]
The law of total covariance and
\eqref{eq:exact_cond_mean}--\eqref{eq:exact_cond_cov} give
\begin{align}
 \Sigma_{t+1}
 &=\frac1N\Expect[\Lambda_{\theta_t}(\mu_t^N)]
 +\operatorname{Cov}(g_{\theta_t}(\mu_t^N)).
 \label{eq:total_cov_exact}
\end{align}
Linearizing the conditional-mean map around $\bar\mu_t$ yields
\[
 g_{\theta_t}(\mu_t^N)
 =g_{\theta_t}(\bar\mu_t)
 +H_t(\mu_t^N-\bar\mu_t)+\text{higher-order terms}.
\]
Under this linearization, $H_t\widehat\Sigma_tH_t^\top$ describes how
existing population fluctuations propagate to the next time step.
Because $\mu_t^N$ fluctuates around $\bar\mu_t$ at rate
$O_p(N^{-1/2})$, the variation of $\Lambda_{\theta_t}(\mu_t^N)$
contributes only a smaller-order term after multiplication by $1/N$;
we therefore evaluate $\Lambda_{\theta_t}$ at $\bar\mu_t$. This gives
\eqref{eq:lyapunov}. %
\end{proof}
\begin{remark}[Aggregate action constraints]
For constraints on actions aggregated across the population, an analogous
first-order variance approximation can be derived for the empirical
occupation measure; see Appendix~\ref{app:occupation_variance}.
\end{remark}

\section{Cantelli Chance-Constraint Certification}
\label{sec:cantelli}

The standard MF approach enforces only the nominal constraints
$\bar\mu_t(\cT)\ge\alpha_r$ and $\bar\mu_t(\cU)\le\beta_u$,
but these mean-field conditions alone do not imply the finite-$N$
chance constraints~\eqref{eq:reach}--\eqref{eq:avoid}.
Write $m_t^N:=\Expect[\mu_t^N]$ for the exact finite-population mean and
$m_t^N(\mathcal R):=c_{\mathcal R}^{\top}m_t^N$ for the expected fraction
of agents in $\mathcal R$.
Here, $c_{\mathcal R}(s)=\mathbf{1}[s\in\mathcal R]$ is the indicator
vector for $\mathcal R$.

We use Cantelli's inequality to derive deterministic sufficient
conditions based on the first two moments, following a construction
similar to that used in chance-constrained
MPC~\cite{farinaScattoliniChanceMPC2016}.
\begin{proposition}[Cantelli Bounds]
\label{prop:cantelli}
For any $t$ and policy $\pi$, the reach bound below holds when
$m_t^N(\cT)\ge\alpha_r$, and the avoid bound holds when
$m_t^N(\cU)\le\beta_u$:
\begin{align}
    \Prob^\pi\!\bigl(\mu_t^N(\cT)\ge\alpha_r\bigr)
    &\ge 1-\frac{\sigma_t^2(\cT)}
                {\sigma_t^2(\cT)+(m_t^N(\cT)-\alpha_r)^2},
    \label{eq:cantelli_reach}
    \\
    \Prob^\pi\!\bigl(\mu_t^N(\cU)\le\beta_u\bigr)
    &\ge 1-\frac{\sigma_t^2(\cU)}
                {\sigma_t^2(\cU)+(\beta_u-m_t^N(\cU))^2},
    \label{eq:cantelli_avoid}
\end{align}
where $\sigma_t^2(\cR)=c_\cR^\top\Sigma_t c_\cR$ denotes the exact
finite-population variance.
Specifically, when the relevant variance is zero, the corresponding
probability bound equals one, and the Cantelli condition reduces to
$m_t^N(\cT)\ge\alpha_r$ for reach or
$m_t^N(\cU)\le\beta_u$ for avoid.
\end{proposition}

\begin{proof}
For the reach bound, let $X=\mu_t^N(\cT)$ and
$d_r=m_t^N(\cT)-\alpha_r\ge0$. If
$\sigma_t^2(\cT)>0$, Cantelli's inequality gives
\[
\Prob\bigl(X-m_t^N(\cT)\le-d_r\bigr)
\le
\frac{\sigma_t^2(\cT)}{\sigma_t^2(\cT)+d_r^2}.
\]
Since $\{X<\alpha_r\}\subseteq
\{X-m_t^N(\cT)\le-d_r\}$, taking complements yields
\eqref{eq:cantelli_reach}. If the variance is zero, $X$ is constant
almost surely and the conclusion follows directly. The bound on the avoid condition is derived analogously.
\end{proof}

\begin{lemma}[Sufficient Cantelli Conditions]
\label{lem:cantelli_sufficient}
Setting $\kappa_r=\sqrt{(1-\delta_r)/\delta_r}$ and
$\kappa_u=\sqrt{(1-\delta_u)/\delta_u}$, the chance
constraints~\eqref{eq:reach}--\eqref{eq:avoid} are implied by:
\begin{align}
    m_{t^*}^N(\cT)-\alpha_r
    &\ge\kappa_r\,\sigma_{t^*}(\cT),
    \label{eq:suff_reach}
    \\
    \beta_u-m_t^N(\cU)
    &\ge\kappa_u\,\sigma_t(\cU),
    \quad \forall t\le t^*.
    \label{eq:suff_avoid}
\end{align}
\end{lemma}

\begin{proof}
For the reach constraint, let
$d_r:=m_{t^*}^N(\cT)-\alpha_r$. Condition~\eqref{eq:suff_reach}
ensures $d_r\ge0$ and
\begin{equation*}
d_r^2\ge \frac{1-\delta_r}{\delta_r}\,
\sigma_{t^*}^2(\cT).
\end{equation*}
Therefore,
\begin{equation*}
\frac{\sigma_{t^*}^2(\cT)}
{\sigma_{t^*}^2(\cT)+d_r^2}\le\delta_r,
\end{equation*}
and substituting this inequality into~\eqref{eq:cantelli_reach}
gives the required probability bound. For the avoid constraint, set
$d_u:=\beta_u-m_t^N(\cU)$. Condition~\eqref{eq:suff_avoid} gives
$d_u\ge0$ and
$d_u^2\ge((1-\delta_u)/\delta_u)\sigma_t^2(\cU)$, so the same
argument applied to~\eqref{eq:cantelli_avoid} yields the avoid bound.
If either variance is zero, the corresponding random variable is
constant almost surely and the conclusion follows directly.
\end{proof}

\begin{corollary}[Robust Finite-$N$ Certificate]
\label{cor:robust_finite_n}
Let $\widehat\Sigma_t$ be the covariance approximation produced by
Proposition~\ref{prop:lyapunov}. Suppose validated deterministic bounds
$b_{t,\cR},e_{t,\cR}\ge0$ are available for the mean and variance
approximation errors, respectively, such that
\begin{equation}
 \left|c_\cR^\top(m_t^N-\bar\mu_t)\right|\le b_{t,\cR},\qquad
 c_\cR^\top\Sigma_t c_\cR
 \le c_\cR^\top\widehat\Sigma_t c_\cR+e_{t,\cR}.
 \label{eq:moment_enclosures}
\end{equation}
Then the exact finite-$N$ chance constraints are implied by
\begin{align}
 c_\cT^\top\bar\mu_{t^*}-b_{t^*,\cT}-\alpha_r
 &\ge \kappa_r\sqrt{c_\cT^\top\widehat\Sigma_{t^*}c_\cT+e_{t^*,\cT}},
 \label{eq:robust_suff_reach}\\
 \beta_u-c_\cU^\top\bar\mu_t-b_{t,\cU}
 &\ge \kappa_u\sqrt{c_\cU^\top\widehat\Sigma_t c_\cU+e_{t,\cU}},
 \quad \forall t\le t^*.
 \label{eq:robust_suff_avoid}
\end{align}
\end{corollary}

\begin{proof}
The first enclosure lower-bounds the exact reach mean and upper-bounds
the exact unsafe mean; the second upper-bounds the relevant exact
variance. Thus \eqref{eq:robust_suff_reach}--\eqref{eq:robust_suff_avoid}
imply the conditions of Lemma~\ref{lem:cantelli_sufficient}.
\end{proof}

Replacing the   moments in~\eqref{eq:suff_reach}--\eqref{eq:suff_avoid}
with $c_{\mathcal R}^{\top}\bar\mu_t(\theta)$ and
 $\widehat{\Sigma}_t( \theta) $,
for $\mathcal R\in\{\mathcal T,\mathcal U\}$, gives  
surrogate constraints. These constraints are nonlinear
in $\theta$ through~\eqref{eq:mf_recurrence} and~\eqref{eq:lyapunov}.

\section{Chance-Constrained Mean-Field Policy Synthesis}
\label{sec:sca}

Using the mean-field and first-order covariance approximations, we
formulate the following approximate synthesis problem:

\begin{subequations}
	\label{eq:surrogate_problem}
	\begin{align}
		\max_ {\theta, t^\ast} \quad &
		J(\theta) =  \sum_{t=0}^{T-1} \sum_{s,a} r(s,a) \bar{\mu}_t(s; \theta) \pi_{\theta}(a|s, \bar{\mu}_t(\theta))  \notag \\
		& =   \sum_{t=0}^{T-1} \sum_{s,a} r(s,a)    z_t(s,a;\theta) \notag 
		\\
		\text{s.t.}\quad & \eqref{eq:robust_suff_avoid},\eqref{eq:robust_suff_reach},
 \eqref{eq:mf_recurrence}, \eqref{eq:lyapunov}. \notag
	\end{align} 
\end{subequations}
Recall that $z_t$ is the occupancy measure. 

Because the problem is nonconvex,  we solve it using the sequential convex approximation (SCA) method.
At each SCA iteration, we linearize the mean-field trajectory and solve
a trust-region SOCP that takes one policy-improvement step subject to
the Cantelli reach-avoid constraints. We first derive the required
Jacobians, then formulate the SOCP subproblem, and finally describe the
trust-region update and the search over candidate reach times.
\subsection{Linearization of the objective function}

Given the policy parameter $\theta^{(k)}$ at iteration $k$, the objective function has the following first-order Taylor expansion:
$$
\begin{aligned}
	& J(\theta^{(k)}+\Delta\theta)\\
	=& \sum_{t=0}^{T-1}\sum_{s,a} r(s,a) z_t(s,a;\theta^{(k)}+\Delta\theta)\\
	\approx & \sum_{t=0}^{T-1}\sum_{s,a} r(s,a)\left[
	z_t(s,a;\theta^{(k)}) + \nabla_\theta z_t(s,a;\theta^{(k)}) \Delta\theta
	\right]
\end{aligned}
$$
Recall that $z_t(s,a;\theta) = \bar \mu_t(s;\theta) \pi_{\theta_t}(a\mid s,\bar \mu_t(\theta))$, which is a function of $\theta$ through both the parameterized policy and $\bar \mu_t(\theta)$.
 Let $\mathbf z_t(\theta)=f_t(\theta,\bar\mu_t(\theta))$, and  
$D_t^{(k)}:=\partial\bar\mu_t/\partial\theta\big|_{\theta^{(k)}}
\in\Real^{|\cS|\times n_\theta}$ denote the Jacobian of the mean-field
density with respect to the policy parameters, with $D_0^{(k)}=0$
because $\bar\mu_0=\mu_0$ is fixed.
Here $n_\theta=T|\mathcal S||\mathcal A|n_{\rm feat}$ is the total
number of policy parameters.

\begin{proposition}
\label{prop:linearization}
The Jacobian of $z_t$ at iterate $\theta^{(k)}$ is
\begin{equation} \left.
	\frac{\partial \mathbf z_t(\theta)}{\partial\theta}
	\right|_{\theta=\theta^{(k)}}
	 :=  G_t^{(k)} = J_t^\theta + J_t^\mu D_t^{(k)},
    \label{eq:G_full}
\end{equation}
where $J_t^\theta:=\partial_\theta f_t\big|_{(\theta^{(k)},\bar\mu_t^{(k)})}$
and $J_t^\mu:=\partial_\mu f_t\big|_{(\theta^{(k)},\bar\mu_t^{(k)})}$, and
$D_t^{(k)}$ propagates forward via
\begin{equation}
    D_{t+1}^{(k)} = F\,G_t^{(k)}, \qquad D_0^{(k)}=0.
    \label{eq:D_recursion}
\end{equation}
Consequently, the linearized occupation measure and mean trajectory are
$\mathbf{z}_t^{\rm lin}(\dtheta)=\mathbf{z}_t^{(k)}+G_t^{(k)}\dtheta$
and $\bar\mu_t^{\rm lin}(\dtheta)=\bar\mu_t^{(k)}+D_t^{(k)}\dtheta$.
\end{proposition}

\begin{proof}
By \eqref{eq:mf_recurrence},
$\bar\mu_{t+1}(\theta)=F\mathbf z_t(\theta)
=Ff_t(\theta,\bar\mu_t(\theta))$.

Since $f_t$ depends on $\theta$ both directly and through $\bar\mu_t(\theta)$, the chain rule with
$\partial_\theta f_t=J_t^\theta$,
$\partial_\mu f_t=J_t^\mu$, $d \bar{\mu}_t /d\theta=D_t^{(k)}$ gives exactly
\eqref{eq:G_full}:
\begin{equation*}
    G_t^{(k)} = \frac{dz_t}{d\theta}\bigg|_{\theta^{(k)}}
    = J_t^\theta + J_t^\mu D_t^{(k)}.
\end{equation*}
Here, $J_t^\theta(s,a) = \bar\mu_t^{(k)}(s)\nabla_\theta\pi_{\theta_t}(a\mid s,\bar\mu_t^{(k)})$
and \begin{multline*}
	J_t^\mu(s,a)[s''] = \pi^{(k)}(a\mid s,\bar\mu_t^{(k)} )\,\delta_{ss''}
	+ \\
	\bar\mu_t^{(k)}(s)\,
	\partial_{\mu(s'')}\pi_\theta(a\mid s,\mu)\big|_{\bar\mu_t^{(k)}}.
	\label{eq:Jmu_coord}
\end{multline*}

The flow matrix
$F$ is constant, $\theta$-independent, so differentiating
$\bar\mu_{t+1}(\theta)=Fz_t(\theta)$ directly gives
$D_{t+1}^{(k)}=F\,dz_t/d\theta=FG_t^{(k)}$, i.e.\ \eqref{eq:D_recursion}.
Substituting the expression for $G_t^{(k)}$ above into this combines the pair
\eqref{eq:G_full}--\eqref{eq:D_recursion} into a single linear
recursion   in $D_t^{(k)}$:
\begin{equation}
    D_{t+1}^{(k)} = FJ_t^\mu D_t^{(k)}+FJ_t^\theta,
    \qquad D_0^{(k)}=0.
    \label{eq:D_decoupled}
\end{equation}
The first term propagates the density sensitivity, while the second
captures the direct effect of the policy parameters.
$G_t^{(k)}$ is then recovered
from $D_t^{(k)}$ via \eqref{eq:G_full}.

For $\theta=\theta^{(k)}+\dtheta$,
\begin{align*}
    z_t(\theta^{(k)}+\dtheta)
    &\approx z_t(\theta^{(k)}) + G_t^{(k)}\dtheta
    = \mathbf z_t^{(k)} + G_t^{(k)}\dtheta,
    \\
    \bar\mu_t(\theta^{(k)}+\dtheta)
    &\approx \bar\mu_t(\theta^{(k)}) + D_t^{(k)}\dtheta
    = \bar\mu_t^{(k)} + D_t^{(k)}\dtheta,
\end{align*}
\end{proof}

\subsection{SOCP Subproblem}
For a finite-$N$ certificate, let $b_{t,\cR}^{(k)}$ and
$e_{t,\cR}^{(k)}$ be validated upper bounds satisfying
\eqref{eq:moment_enclosures} throughout the current trust region.
Linearising the mean and freezing these bounds gives the $k$-th SCA
subproblem:
\begin{subequations}
\label{eq:sca_subproblem}
\begin{align}
    \max_{\dtheta,\,\xi\ge0}\;&
    \sum_{t=0}^{T-1}\mathbf{r}^\top
    (\mathbf{z}_t^{(k)}{+}G_t^{(k)}\dtheta)
    - \lambda\textstyle\sum_t\xi_t
    \label{eq:sca_obj}
    \\
    \text{s.t.}\;&
    c_\cT^\top\bar\mu_{t^*}^{\rm lin}(\dtheta)
    \ge\alpha_r+b_{t^*,\cT}^{(k)}
    +\kappa_r\widehat\sigma_{t^*}^{\cT,(k)},
    \label{eq:sca_reach}
    \\
    &c_\cU^\top\bar\mu_t^{\rm lin}(\dtheta)-\xi_t
    \le\beta_u-b_{t,\cU}^{(k)}
    -\kappa_u\widehat\sigma_t^{\cU,(k)},
    \;\forall t\le t^*,
    \label{eq:sca_avoid}
    \\
    &\|\dtheta\|_2\le\rho^{(k)},
    \label{eq:sca_tr}
\end{align}
\end{subequations}
where $\xi_t$ are slack variables, $\mathbf{r}$ is the reward function written in the vector form, $t^*$ is fixed (see below), $\rho^{(k)}>0$ is the trust-region
radius,
\begin{equation}
    \bar\mu_t^{\rm lin}(\dtheta) := \bar\mu_t^{(k)} + D_t^{(k)}\dtheta
    \label{eq:mu_lin_recall}
\end{equation}
is the linearized mean trajectory of Proposition~\ref{prop:linearization},
and
\begin{equation}
    \widehat\sigma_t^{\cR,(k)} :=
    \sqrt{\max\bigl(c_\cR^\top\widehat\Sigma_t^{(k)}c_\cR
    +e_{t,\cR}^{(k)},\,0\bigr)},
    \qquad \cR\in\{\cT,\cU\},
    \label{eq:sigma_frozen}
\end{equation}
is a certified upper bound on the Cantelli standard deviation of
region $\cR$. Here $\widehat\Sigma_t^{(k)}$ is the first-order
covariance recursion evaluated at the current iterate. The resulting
constraints are the robust conditions
\eqref{eq:robust_suff_reach}--\eqref{eq:robust_suff_avoid} with a
linearized mean.  
Problem~\eqref{eq:sca_subproblem} is an SOCP (the objective and
constraints are affine in $\dtheta$; the trust region is a
second-order cone constraint), solvable in polynomial time.

\paragraph{Fixed-covariance surrogate}
The mean trajectory is linearized as $\bar\mu_t^{\rm lin}(\dtheta)$,
whereas the covariance surrogate is held at $\widehat\Sigma_t^{(k)}$.
This is because $\widehat\Sigma_t(\theta)$ depends on $\theta$ through
the nonlinear Lyapunov recursion~\eqref{eq:lyapunov}, which contains
the term $H_t(\theta)\widehat\Sigma_t(\theta)H_t(\theta)^\top$.
Freezing $\widehat\sigma_t^{\cR,(k)}$ keeps the constraints affine in
$\dtheta$ and preserves the SOCP form of the subproblem.
The trust region keeps this approximation local, and the covariance is
recomputed after every accepted step. By Proposition~\ref{prop:lyapunov}, the change in covariance over a gradient step is
typically $O(\|\dtheta\|/N)$.  Such a small error can be negligible or absorbed into the
error    $e_{t,\cR}$ of
Corollary~\ref{cor:robust_finite_n}.

\paragraph{Trust-Region Update}
For each SCA run, $t^*$ is fixed to one candidate reach time. 
Since a step that improves $J$ can violate the robust constraints after
linearization, we accept or reject it using an $L_1$-penalty merit
function evaluated at the nonlinear robust margins,
\begin{equation}
	\begin{aligned}
	M(\theta) := J(\theta) - \lambda_M\Bigl[&
	\max\bigl(0,-g_{\rm reach}(\theta)\bigr)\\[-1mm]
	&+ \max\bigl(0,-g_{\rm avoid}(\theta)\bigr)\Bigr],
	\end{aligned}
	\label{eq:merit}
\end{equation}
where we define the robust standard-deviation margin
$q_{t,\cR}(\theta):=\sqrt{c_\cR^\top\widehat\Sigma_t(\theta)c_\cR
+e_{t,\cR}(\theta)}$. Then
\begin{align*}
g_{\rm reach}(\theta)
&=c_\cT^\top\bar\mu_{t^*}(\theta)-\alpha_r-b_{t^*,\cT}(\theta)
-\kappa_r q_{t^*,\cT}(\theta),\\
g_{\rm avoid}(\theta)
&=\min_{t\le t^*}\Bigl[\beta_u-c_\cU^\top\bar\mu_t(\theta)
-b_{t,\cU}(\theta)\\[-1mm]
&\qquad-\kappa_u q_{t,\cU}(\theta)\Bigr].
\end{align*}
The bounds $b_{t,\cR}(\theta)$ and $e_{t,\cR}(\theta)$ are recomputed
or validated at the candidate policy; $\lambda_M>0$ penalizes
constraint violation.
If the SOCP is feasible, set
$\theta'=\theta^{(k)}+\dtheta^{(k)}$ and evaluate the candidate using
the nonlinear merit function. The trust-region radius is then updated as:
\begin{equation}
	\rho^{(k+1)}=
	\begin{cases}
		\min(\beta_{\rm inc}\rho^{(k)},\rho_{\rm max})
		& \text{if }M(\theta')\ge M(\theta^{(k)})-\varepsilon,
		\\
		\beta_{\rm dec}\rho^{(k)} & \text{otherwise.}
	\end{cases}
	\label{eq:trust_update}
\end{equation}
where $\beta_{\rm dec}\in(0,1)$, $\beta_{\rm inc}>1$, and
$\varepsilon\ge0$ is the acceptance tolerance.
Candidates that fail the merit-function acceptance test are rejected
and the trust-region radius is reduced.

If instead the SOCP~\eqref{eq:sca_subproblem} is infeasible, the trust region
may be too small to reach the required constraint shifts. We therefore
increase $\rho^{(k)}$ up to $\rho_{\rm max}$. If infeasibility persists
at $\rho_{\rm max}$, we terminate the current SCA run and proceed to
the next candidate $t^*$. 

\paragraph{Choose candidate reach time}
The existential reach constraint requires only that
\eqref{eq:sca_reach} holds at \emph{some} $t^*$, not a specific terminal time $T$.
Given an initial warm-start policy, we enumerate the
top-$K_{t^*}$ candidate time steps by their initial reach margins
$c_\cT^\top\bar\mu_t^{(0)}-\alpha_r-b_{t,\cT}^{(0)}
-\kappa_r\widehat\sigma_t^{\cT,(0)}$
(descending). We run SCA for each candidate in turn and stop when both
margins in the merit function~\eqref{eq:merit} satisfy
$r_m:=g_{\rm reach}\ge0$ and $a_m:=g_{\rm avoid}\ge0$. If no candidate
satisfies both margins, we return the infeasible result with the largest
$M(\theta)$ across all $K_{t^*}$ candidates. Because the search stops
at the first candidate satisfying both margins, rather than comparing
all $K_{t^*}$ candidates, the returned policy may be suboptimal in reward.

\RestyleAlgo{ruled}
\begin{algorithm}[t]
\footnotesize
\caption{SCA for Reach-Avoid Swarm Control}
\label{alg:sca}
\KwIn{$\theta^{(0)}$, $N$, $K$, $K_{t^*}$,
      $\rho_0$, $\beta_{\rm inc}$, $\beta_{\rm dec}$,
      $\rho_{\rm max}$, $\lambda$, $\lambda_M$, $\varepsilon$}
\KwOut{$\theta^*$, $\bar\mu^*$, $\widehat\Sigma^*$, $r_m^*$, $a_m^*$}
Compute $(\bar\mu_t^{(0)},\widehat\Sigma_t^{(0)})$ via
forward sim.\ and Lyapunov recursion at $\theta^{(0)}$\;
Compute or validate $b_{t,\cR}^{(0)},e_{t,\cR}^{(0)}$ for the initial
trust-region radius $\rho_0$\;
Rank $t\in[T]$ by $c_\cT^\top\bar\mu_t^{(0)}-\alpha_r
-b_{t,\cT}^{(0)}-\kappa_r\widehat\sigma_t^{\cT,(0)}$ descending;
set $\mathcal{T}^*\leftarrow$ top-$K_{t^*}$ candidates\;
\ForEach{$t^*\in\mathcal{T}^*$}{
    $\theta\leftarrow\theta^{(0)}$\;
    $\rho\leftarrow\rho_0$\;
    \For{$k=0,\ldots,K-1$}{
        Compute $(\bar\mu_t^{(k)},\widehat\Sigma_t^{(k)})$\;
        Compute or validate $b_{t,\cR}^{(k)},e_{t,\cR}^{(k)}$ for the
        current trust region\;
        Compute $G_t^{(k)}, D_t^{(k)}$ via
        Prop.~\ref{prop:linearization}\;
        Solve SOCP~\eqref{eq:sca_subproblem}
        $\to\dtheta^{(k)}$\;
        \If{infeasible}{
            \If{$\rho=\rho_{\rm max}$}{
                \textbf{break}\tcp*{proceed to the next $t^*$}
            }
            $\rho\leftarrow\min(\beta_{\rm inc}\rho,\rho_{\rm max})$;
            \textbf{continue}\;
        }
        $\theta'\leftarrow\theta+\dtheta^{(k)}$\;
        \eIf{$M(\theta')\ge M(\theta)-\varepsilon$}{
            $\theta\leftarrow\theta'$\;
            $\rho\leftarrow\min(\beta_{\rm inc}\rho,\rho_{\rm max})$\;
        }{
            $\rho\leftarrow\beta_{\rm dec}\rho$\;
        }
    }
    Evaluate $r_m,a_m,M(\theta)$\;
    \If{$r_m\ge0$ and $a_m\ge0$}{
        Store the candidate\;
        \textbf{break}\tcp*{first candidate satisfying both margins}
    }
    Store if $M(\theta)$ is the largest infeasible value so far\;
}
\Return{stored $(\theta^*,\bar\mu^*,\widehat\Sigma^*,r_m^*,a_m^*)$}
\end{algorithm}

\section{Experimental Validation}
\label{sec:experiments}
We validate the algorithm using two case studies.\footnote{Code is
available at \url{https://github.com/jiefu2017/mfmdp-swarm-control}.}
As a baseline, we use a mean-field LP that solves the mean-field relaxation of
Problem~\ref{prob:main} directly over the occupation measure
$\mathbf z_t(s,a)\ge0$, rather than over policy parameters $\theta$:
\begin{subequations}
	\label{eq:mflp}
	\begin{align}
		\max_{\mathbf z_0,\ldots,\mathbf z_{T-1}\ge0}\;&
		\sum_{t=0}^{T-1}\mathbf r^\top\mathbf z_t
		\\
		\text{s.t.}\;&
		\bar\mu_{t+1} = F\mathbf z_t,\quad \bar\mu_0=\mu_0,
		\\
		& c_\cT^\top\bar\mu_{t^*}\ge\alpha_r,
		\quad c_\cU^\top\bar\mu_t\le\beta_u\;\;\forall t\le t^*,
	\end{align}
\end{subequations}
enforcing only the \emph{mean} constraints.
Since $t^*$ is existential, \eqref{eq:mflp} is solved once for each
candidate $t^*\in[T]$ and the best feasible reward is kept.
Because the objective and constraints are linear in $\mathbf z_t$,
\eqref{eq:mflp} is an LP and is solved exactly.
The policy is recovered as
$\pi_t(a\mid s)=z_t^*(s,a)/\bar\mu_t^*(s)$ (uniform where
 $\bar\mu_t^*(s)=0$). %

\vspace{-0.5em}
\subsection{Gridworld Environment}

Our first case study is a $6\times6$ discrete gridworld
(Fig.~\ref{fig:gridworld}, $|\cS|=36$,
$|\cA|=5$: up/down/left/right/stay, $T=20$).
Transitions: intended action with probability $0.9$; otherwise a
uniformly random \emph{other} action executes instead (probability
$0.1/4$ each). Moves into a wall cell, or off the grid, leave the agent
in place.
Start: $(5,0)$.
A wall at $x=2$ separates the swarm from three reach targets
$\cT=\{(0,0),(0,3),(0,4)\}$ (reward $0.5$ each) and one reward-only
target $(0,5)$ (reward $1.0$, not in $\cT$): higher reward, but not
reach-designated, mirroring the reward/safety tension of
the reach-avoid setting considered here.
The wall has three physically impassable cells $(2,0),(2,2),(2,4)$ and
three openings $(2,1),(2,3),(2,5)$; every route from the start to any
target must pass through one of these openings.

\begin{figure}[t]
\centering
\includegraphics[width=0.75\columnwidth]{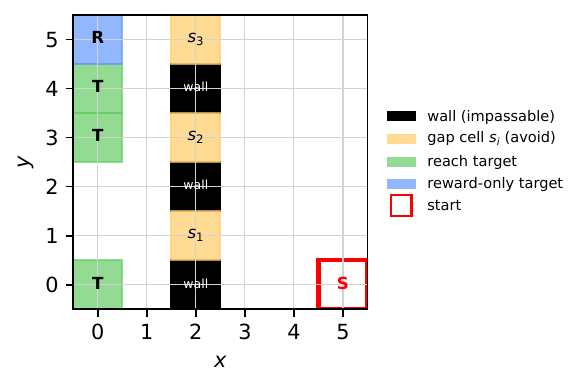}
\vspace{-3ex}
\caption{The $6\times6$ gridworld environment. Start (red, $S$) at
$(5,0)$; wall at $x=2$ with three impassable cells (black) and three
open gap cells $s_1,s_2,s_3$ (orange) that every route must cross;
three reach targets (green, $T$) and one reward-only target (blue,
$R$), both past the wall.}
\label{fig:gridworld}
\vspace{-3ex}
\end{figure}

Rather than one   avoid constraint, we consider
  three independent per-cell instances of the avoid clause
in~\eqref{eq:mflp} (and its Cantelli counterpart~\eqref{eq:sca_avoid}),
 tracking each cell's \emph{instantaneous} occupancy
(congestion): for opening $i$,
$\Prob(\mu_t^N(s_i)>0.30)\le0.40$, i.e.\ $\beta_{u,i}=0.30$,
$\delta_{u,i}=0.40$
($\kappa_{u,i}=\sqrt{(1-\delta_{u,i})/\delta_{u,i}}=\sqrt{1.5}\approx1.22$),
which means the probability that more than $30\%$ of agents occupy
the same opening simultaneously is at most $40\%$. 
  Reach parameters: $\alpha_r=0.3$, $\delta_r=0.1$ ($\kappa_r=3$), which means with probability at least $90\%$, at least $30\%$ of the agents reach the target.

Table~\ref{tab:results} reports the surrogate Cantelli margins $r_m,a_m$
at the $t^*$ selected by each method, together with Monte Carlo estimates
over 500 trials (fleet sizes $N\in\{10,20,50,100\}$), using
Algorithm~\ref{alg:sca} with the LP-based warm start.
$\hat P_{\rm v}$ reports the \emph{maximum} over the three per-cell
violation rates. The target values are $\hat P_{\rm reach}\ge 0.90$ and $\hat P_{\rm v}\le0.40$.

\begin{table}[t]
\centering
\caption{Gridworld results. $r_m$, $a_m$: surrogate Cantelli margins
($\ge0$ indicates that the surrogate condition is satisfied).
$\hat J$: empirical mean total reward. $\hat P_{\rm r}$: empirical reach
probability. $\hat P_{\rm v}$: maximum empirical per-cell avoid
violation rate. Emp.\ met? indicates whether both empirical reach
and avoid requirements are met.}
\label{tab:results}
\renewcommand{\arraystretch}{1.2}
\setlength{\tabcolsep}{3.5pt}
\resizebox{\columnwidth}{!}{%
\begin{tabular}{ccccc|cccc}
\toprule
& \multicolumn{4}{c|}{\textbf{MF-LP (baseline)}}
& \multicolumn{4}{c}{\textbf{SCA (proposed)}} \\
$N$ & $\hat J$ & $\hat P_{\rm r}$ & $\hat P_{\rm v}$
    & emp.\ met?
    & $\hat J$ & $\hat P_{\rm r}$ & $\hat P_{\rm v}$
    & $r_m$/$a_m$ \\
\midrule
 10 & 9.39 & 0.91$^\checkmark$ & 0.38$^\checkmark$ & $\checkmark$
    & 9.38 & 1.00$^\checkmark$ & 0.40$^\checkmark$
    & $-.108^\times/-.194^\times$ \\
 20 & 9.38 & 0.95$^\checkmark$ & 0.39$^\checkmark$ & $\checkmark$
    & 8.97 & 1.00$^\checkmark$ & 0.12$^\checkmark$
    & $+.029/+.000$ \\
 50 & 9.37 & 1.00$^\checkmark$ & 0.44$^\times$ & $\times$
    & 9.23 & 1.00$^\checkmark$ & 0.08$^\checkmark$
    & $+.000/+.000$ \\
100 & 9.38 & 1.00$^\checkmark$ & 0.47$^\times$ & $\times$
    & 9.28 & 1.00$^\checkmark$ & 0.12$^\checkmark$
    & $+.090/+.000$ \\
\bottomrule
\end{tabular}}
\end{table}

\textbf{Key findings.}
(1)~SCA achieves positive Cantelli margins and stays  
under budget ($\hat P_{\rm v}\le0.12$) for $N\ge20$.
At $N=10$, the \emph{empirical} reach rate is $100\%$ and the
\emph{empirical} per-cell violation rate is also $40\%$. However,
$r_m=-0.108$ and $a_m=-0.194$ mean that the surrogate Cantelli
conditions are not satisfied. This reflects the conservatism of the
distribution-free Cantelli bound: the empirical trials meet both
requirements even though the surrogate test does not.

In comparison, the MF-LP violation rate does not decrease with $N$,
rising from $0.38$ at $N=10$ to $0.47$ at $N=100$. This is because the MF-LP's
mean trajectory sits exactly at the boundary,
$\bar\mu_t(s_i)=\beta_{u,i}=0.30$.
By the central limit theorem, the occupancy $\mu_t^N(s_i)$ is
approximately Gaussian for large $N$, and because a Gaussian is
symmetric about its mean, $\Prob(\mu_t^N(s_i) > \beta_{u,i})
\approx 50\%$.
Lastly, for $N\ge20$, the SCA reward is within $5\%$ of the MF-LP
reward: the three reach
targets give SCA   freedom to redistribute reward-collection time
without sacrificing much of the safety margin.
\vspace{-0.5em}
\subsection{EV-Charging Aggregator}
\label{sec:ev-experiment}

Our second case study is an EV-charging application in which the avoid constraint is a transformer
capacity limit on the \emph{aggregate action}, i.e.\ on the occupation
measure $\mathbf z_t^N$ itself.

\textbf{Environment.} States $\cS=\{0,\ldots,10\}$ are state-of-charge
(SOC) bins in $10\%$ increments; the target set is
$\cT=\{8,9,10\}$ (SOC $\ge80\%$). Three actions: \emph{fast} charging
(reward $3$; advances SOC by $2$ bins w.p.\ $0.6$, else $1$ bin),
\emph{slow} charging (reward $1$; advances by $1$ bin w.p.\ $0.9$,
else unchanged), and \emph{idle}/no-charging (reward $0$; SOC
unchanged deterministically).
SOC $10$ is absorbing. All agents start at SOC $0$.
The policy feature vector is
$\phi(s,\mu)=[1,\,s/10,\,\mu(\cT)]$: a bias, normalized SOC, and a
global density-feedback feature reporting how much of the fleet has
already reached the target.
Rather than bounding how many agents occupy an unsafe \emph{state},
the operational constraint here is that, with probability at least $1-\delta_u$, no more than a fraction
$\beta_u$ of the fleet may fast-charge \emph{simultaneously}, at any
time. Define the empirical occupation measure as
$z_t^N(s,a):=N^{-1}\sum_{i=1}^N\mathbf 1[X_t^i=s,A_t^i=a]$, the
fraction of agents in state $s$ taking action $a$ at time $t$. With
$c_F(s,a)=\mathbf 1[a={\rm fast}]$,
\begin{equation}
    \Prob^\pi\!\bigl(c_F^\top\mathbf z_t^N\le\beta_u\bigr)
    \ge 1-\delta_u,\quad\forall t.
    \label{eq:ev_avoid}
\end{equation} %
For the aggregate-action constraint~\eqref{eq:ev_avoid}, we use the
occupation-measure variance approximation in
Appendix~\ref{app:occupation_variance}, with $c=c_F$. The surrogate
Cantelli condition and SCA subproblem then use $G_t^{(k)}$, the Jacobian
of $\mathbf z_t$, for this constrained quantity. The following
parameters are used:
$T=12$, $\alpha_r=0.8$, $\delta_r=0.1$
($\kappa_r=3$), $\beta_u=0.5$, $\delta_u=0.3$ ($\kappa_u=1.53$).
Unlike in the gridworld case study, the reach time is fixed at
$t^*=T=12$ because the requirement is that at least $80\%$ of the fleet
attain an SOC of at least $80\%$ by the terminal time. The
avoid/congestion constraint
is then active over the entire charging window.
\begin{table}[t]
\centering
\caption{EV-charging results. $r_m$, $a_m$: surrogate Cantelli margins
($\ge0$ indicates that the surrogate condition is satisfied), with
$a_m$ computed using~\eqref{eq:cov_z}.
$\hat J$: empirical mean total reward. $\hat P_{\rm r}$: empirical
reach probability. $\hat P_{\rm v}$:
maximum empirical per-timestep marginal avoid violation rate.
Emp.\ met? indicates whether both empirical requirements
($\hat P_{\rm r}\ge0.90$ and $\hat P_{\rm v}\le0.30$) are met.}
\label{tab:ev-results}
\renewcommand{\arraystretch}{1.2}
\setlength{\tabcolsep}{3.5pt}
\resizebox{\columnwidth}{!}{%
\begin{tabular}{ccccc|cccc}
\toprule
& \multicolumn{4}{c|}{\textbf{MF-LP (baseline)}}
& \multicolumn{4}{c}{\textbf{SCA (proposed)}} \\
$N$ & $\hat J$ & $\hat P_{\rm r}$ & $\hat P_{\rm v}$
    & emp.\ met?
    & $\hat J$ & $\hat P_{\rm r}$ & $\hat P_{\rm v}$
    & $r_m$/$a_m$ \\
\midrule
 20 & 19.11 & 1.00$^\checkmark$ & 0.45$^\times$ & $\times$
    & 17.62 & 1.00$^\checkmark$ & 0.27$^\checkmark$
    & $+.138^\checkmark/{-}.134^\times$ \\
 50 & 19.09 & 1.00$^\checkmark$ & 0.49$^\times$ & $\times$
    & 17.56 & 1.00$^\checkmark$ & 0.06$^\checkmark$
    & $+.196^\checkmark/{+}.000^\checkmark$ \\
100 & 19.09 & 1.00$^\checkmark$ & 0.52$^\times$ & $\times$
    & 17.99 & 1.00$^\checkmark$ & 0.08$^\checkmark$
    & $+.199^\checkmark/{+}.000^\checkmark$ \\
200 & 19.09 & 1.00$^\checkmark$ & 0.51$^\times$ & $\times$
    & 18.31 & 1.00$^\checkmark$ & 0.07$^\checkmark$
    & $+.200^\checkmark/{+}.000^\checkmark$ \\
\bottomrule
\end{tabular}}
	\vspace{-3ex}
\end{table}

\textbf{Key findings.} 
(1)~The MF-LP maximum per-time-step violation rate remains near the
CLT-predicted $50\%$ level ($0.45$--$0.52$) across all tested fleet sizes,
which is a property of the LP's mean-only formulation.
(2)~SCA uses the occupation-measure variance approximation
in~\eqref{eq:cov_z}. It satisfies the surrogate conditions for
$N\ge50$
($a_m\approx 0$, $\hat P_{\rm v}\le0.08$ at $N=50,100,200$). As $N$
grows, the aggregate charging fraction concentrates around its mean,
reducing the uncertainty margin required by the surrogate constraint.
For $N\ge50$, the observed violation rate remains at or below $0.08$.
Figure~\ref{fig:ev-idle} illustrates the corresponding mean charging
profiles: the MF-LP fast-charging fraction sits at the $\beta_u=0.5$
limit, whereas the SCA fraction remains around $0.42$. We also observe
that SCA's slow-charging fraction is
higher than MF-LP's, especially at earlier timesteps: because SCA's
average fast-charging fraction is lower, it must rely more on
slow-charging throughout to still satisfy the reach constraint.

Lastly, as an ablation study, we   tested replacing the scalar feature $\mu(\cT)$ with the
full density vector $\mu_t\in\Real^{|\cS|}$ (raising $n_{\rm feat}$
from $3$ to $2+|\cS|=13$, and $\theta$ from $1{,}188$ to $5{,}148$
parameters). The two feature designs produced similar results, so we
omit the full-density results. %

\textbf{Computational performance.} On an Apple M1 Pro with 16~GB RAM,
the gridworld policy has 18{,}000 parameters and uses at most 25 SCA
iterations for each of up to three candidate reach times. Including SCA
and the 500-trial evaluations of both methods, the runtime per fleet
size was 53--174~s. The EV policy has 1{,}188 parameters and uses one
fixed reach time; the corresponding runtime was 12--15~s. The
experiments used Python~3.14, JAX~0.11, and CVXPY~1.9.

\begin{figure}[t]
	\centering
	\includegraphics[width=0.85\columnwidth]{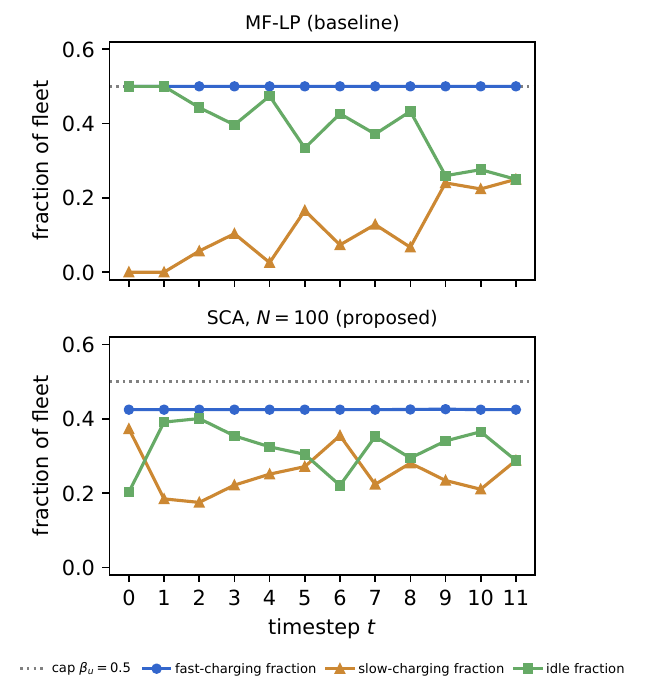}
		\vspace{-3ex}
	\caption{Aggregate fast-charging, slow-charging, and idle fractions
		over time (summing to 1 at every $t$), MF-LP's optimal solution (top)
		versus the SCA solution at $N=100$ (bottom). MF-LP saturates
		the $\beta_u=0.5$ cap at every timestep and avoids slow-charging
		until the departure deadline forces it ($t=9$--$11$); SCA holds a
		margin below the cap throughout.}
	\label{fig:ev-idle}
	\vspace{-3ex}
\end{figure}

\section{Conclusion}
\label{sec:conclusion}

We presented a policy synthesis framework for chance-constrained
reach-avoid control of finite-$N$ swarms on discrete MDPs. The proposed method optimizes 
a density-feedback policy via Sequential Convex Approximation over a
sequence of SOCPs for constrained policy gradient updates. In both case studies, the MF-LP baseline's
mean-only solution lies on the active constraint boundary, causing its
violation probability to approach approximately $50\%$ rather than
vanish as $N$ increases. The standard-deviation penalties in the SCA
constraints decrease at rate $O(1/\sqrt{N})$, reducing conservatism for
larger populations.  Future work includes extending
this approach to true mean-field-type MDPs, where the transition
kernel and reward function are coupled to the population's empirical density, and
to chance-constrained finite-N MDPs with common noise. 

\appendix[Variance of Aggregate Occupation Measures]
\label{app:occupation_variance}

For a linear statistic $Y_t=c^\top\mathbf z_t^N$ of the empirical
occupation measure, where $c\in\Real^{|\cS||\cA|}$, define
\begin{align*}
 h_c(\mu)&:=\sum_{s,a}c(s,a)\pi_{\theta_t}(a\mid s,\mu)\mu(s),\\
 v_t(s,c)&:={\rm Var}_{a\sim\pi_{\theta_t}(\cdot\mid s,\bar\mu_t)}[c(s,a)].
\end{align*}
Conditioned on $\mu_t^N$, agents at each state choose actions
independently, so their action counts have multinomial distributions.
Applying the law of total variance and linearizing the conditional mean
$h_c$ around $\bar\mu_t$ gives
\begin{equation}
 {\rm Var}[Y_t]\approx
 \frac1N\sum_s\bar\mu_t(s)v_t(s,c)
 +\nabla h_c(\bar\mu_t)^\top\widehat\Sigma_t\nabla h_c(\bar\mu_t),
 \label{eq:cov_z}
\end{equation}
where $\widehat\Sigma_t$ is the state-density covariance approximation
from Proposition~\ref{prop:lyapunov}. The first term is the conditional
action-sampling variance; the second propagates state-density
fluctuations through the density-dependent conditional mean. Define
$\widehat\sigma_t(c)$ as the square root of the right-hand side of
\eqref{eq:cov_z}. For the EV-charging example, $c=c_F$ selects the
fraction of agents taking the fast-charging action.

\bibliographystyle{IEEEtran}

\end{document}